\documentclass{amsart}

\usepackage{afterpage} 

\usepackage{amssymb, amsmath, amsthm, amsfonts}

\usepackage{array} 
\newcolumntype{P}[1]{>{\centering\arraybackslash}p{#1}}

\usepackage{bm}

\usepackage{caption}
\usepackage{changepage}

\usepackage{enumerate} 

\usepackage{epigraph}

\usepackage{faktor} 

\usepackage{float}

\usepackage{graphicx}

\usepackage[utf8]{inputenc}

\usepackage{longtable}

\usepackage{mathtools}

\usepackage{soul}
\usepackage{subcaption}
\usepackage{tasks}

\usepackage{tensor} 

\usepackage{tikz}
\usepackage{tikz-cd}
\usetikzlibrary{shapes.geometric}
\usetikzlibrary{cd}
\usetikzlibrary{graphs,decorations.pathmorphing,decorations.markings}
\usetikzlibrary{chains,fit,positioning}

\usepackage{dynkin-diagrams}

\usepackage{wasysym}

\usepackage{xcolor}
\definecolor{lime}{HTML}{A6CE39}
\DeclareRobustCommand{\orcidicon}{%
	\begin{tikzpicture}
	\draw[lime, fill=lime] (0,0) 
	circle [radius=0.16] 
	node[white] {{\fontfamily{qag}\selectfont \tiny ID}};
	\draw[white, fill=white] (-0.0625,0.095) 
	circle [radius=0.007];
	\end{tikzpicture}
	\hspace{-2mm}
}
\foreach \x in {A, ..., Z}{%
	\expandafter\xdef\csname orcid\x\endcsname{\noexpand\href{https://orcid.org/\csname orcidauthor\x\endcsname}{\noexpand\orcidicon}}
}

\newcommand{\abs}[1]{\left \lvert #1 \right \rvert}
\newcommand{\pangle}[1]{\left\langle #1 \right\rangle}
\newcommand{\pbrace}[1]{\left\{ #1 \right\} }
\newcommand{\pround}[1]{\left( #1 \right)}
\newcommand{\psquare}[1]{\left[ #1 \right]}
\newcommand{\floor}[1]{\left \lfloor #1 \right \rfloor}

\newtheorem{theorem}{Theorem}[section]
\newtheorem{lemma}[theorem]{Lemma}

\newtheorem{prop}[theorem]{Proposition}

\newtheorem{remark}[theorem]{Remark}

\newcommand{\secmath}[1]{\texorpdfstring{$ #1 $}{T}}

\DeclareMathOperator{\Ad}{Ad}
\DeclareMathOperator{\Aut}{Aut}

\DeclareMathOperator{\diag}{diag}

\DeclareMathOperator{\id}{Id}

\DeclareMathOperator{\Jac}{Jac}
\DeclareMathOperator{\Ker}{Ker}

\DeclareMathOperator{\Orth}{O}

\DeclareMathOperator{\Pic}{Pic}

\DeclareMathOperator{\PP}{\mathbb{P}^1}

\DeclareMathOperator{\Res}{Res}

\DeclareMathOperator{\SO}{SO}
\DeclareMathOperator{\OO}{O}

\DeclareMathOperator{\SU}{SU}

\DeclareMathOperator{\Tr}{Tr}

\usepackage[colorlinks = true,
            linkcolor = blue,
            urlcolor  = red,
            citecolor = red,
            anchorcolor = blue,
            urlbordercolor= red ]{hyperref}
\usepackage{blkarray,booktabs,bigstrut}
\newtheorem*{theorem*}{Theorem}

\newcommand{\charge}{\ensuremath{k}}

\title{Dihedrally Symmetric Monopoles and Affine Toda Equations}

\author[H. W. Braden]{H. W. Braden \orcidA{}}
\address{
School of Mathematics and Maxwell Institute for Mathematical Sciences\\ The University of Edinburgh\\ 
Edinburgh EH9 3FD, Scotland, U.K.
}
\email{hwb@ed.ac.uk}

\author[Linden Disney-Hogg]{Linden Disney-Hogg \orcidB{}}
\address{
School of Mathematics \\ The University of Leeds \\ 
Leeds LS2 9JT, U.K.
}
\email{a.l.disney-hogg@leeds.ac.uk}

\date{July 2024}

\begin{document}

%%%%%%%%%%%%%%%%%%%%%%%%%%%%%%%%%%%%%%%%%%%%
%%%%%%%%%%%%%%%%%%%%%%%%%%%%%%%%%%%%%%%%%%%%
\begin{abstract}We show that any $SU(2)$ BPS monopole of charge $\charge$ with rotational spatial dihedral
symmetry is gauge equivalent to the Nahm data obtained from affine Toda equations of $C_l\sp{(1)}$ type when $\charge=2l$ or $A_{2(l-1)}\sp{(2)}$ type when $\charge=2l-1$. 

\end{abstract}

\maketitle

\section{Introduction}

Nahm's equations originally arose in the study of magnetic monopoles as
a symmetry reduction of the anti-self-dual equation
\begin{align} \label{eq:ASDequation}
\ast F_A = - F_A,
\end{align}
for a principal $G$-bundle $E$ and connection $A$ over a 4-dimensional manifold.
By taking $\Gamma=\mathbb{R}^3$ to be the closed subgroup of $\mathbb{R}^4$, the group of translations, and requiring the bundle and connection on $\mathbb{R}^4$ to be invariant under $\Gamma$ this gives us a (reduced) anti-self-dual pair $(A,E)$ over $X=\mathbb{R}^4 / \Gamma$ that satisfies Nahm's equations,
 the system of ordinary differential equations
\begin{align} \label{eq:NahmT0}
\begin{split} \frac{dT_1}{ds}&=[T_4,T_1]+[T_2,T_3], \\
\frac{dT_2}{ds}&=[T_4,T_2]+[T_3,T_1], \\
\frac{dT_3}{ds}&=[T_4,T_3]+[T_1,T_2]. \\
\end{split}
\end{align}
Here $T_1(s), T_2(s), T_3(s), T_4(s)$ are antihermitian $\charge\times \charge$ matrix-valued functions and the real parameter $s$ lies in an interval, or union of intervals, depending on the problem. A gauge transform acts as 
\begin{align*}
    T_i &\mapsto UT_i U^{-1}, \quad i=1,2,3, \\
    T_4 &\mapsto U T_4 U^{-1} - \frac{dU}{ds} U^{-1},
\end{align*}
where $U(s)$ is a unitary $\charge \times \charge$ matrix on the same interval(s). By a choice of gauge transform we may set $T_4=0$, and the residual gauge freedom acts via constant $U$. 
The boundary condition for the $T_*$'s also depend critically on the problem. We shall refer to \emph{Nahm data} as solutions to (\ref{eq:NahmT0}) satisfying these (as yet to be specified) boundary conditions for the relevant interval, or union of intervals.
If one takes the dual space $X^\ast=\mathbb{R}^4 / \Gamma^\ast$, where $\Gamma^\ast$ is isomorphic to $\mathbb{R}$, the reduced anti-self-dual pair $(\hat{A},\hat{E})$ is a BPS monopole. This correspondence between equations on
$X$ and $X\sp\ast$ is known as the Nahm transform,
an operation here that transforms solutions of Nahm's equations into monopoles, and vice versa \cite{Nahm1980, Nahm1982}.
The natural appearance of Nahm's equations in such diverse areas as hyperk\"ahler geometry and representation theory has led to their own
independent study. 

Of particular interest to us is the connection between integrable systems and Nahm's equations. Upon setting (with ${T_i}\sp\dagger=-T_i$, $T_4\sp\dagger =-T_4$)
\begin{equation}\label{nahmlax}
\begin{split}
\alpha&=T_4+i T_3,\ \ \beta = T_1+iT_2,
\\
L&=L(\zeta):=\beta -(\alpha+\alpha\sp\dagger)\zeta-\beta\sp\dagger \zeta^2, \ \
M=M(\zeta):=-\alpha-\beta\sp\dagger \zeta,
\end{split}
\end{equation}
for $\zeta \in \mathbb{C}$ one finds\footnote{Throughout we use the standard notation of $\dot{\phantom{T}}$ to mean $\frac{d}{ds}$.}
\begin{equation}
\begin{split}\label{integrability}
\dot{{T}_i} =[T_4,T_i]+\frac12\sum_{j,k=1}^3\epsilon_{ijk}[T_j(s),T_k(s)]
&\Longleftrightarrow
\dot L=[L,M]\\
%&\Longleftrightarrow\quad
%\left\{\begin{aligned}
%\left[\dfrac{d }{dz}-\alpha,\beta\right]&=0,\\
%\dfrac{d (\alpha+\alpha\sp\dagger)}{dz}&=[\alpha,\alpha\sp\dagger]+[\beta,\beta\sp\dagger].
%&\end{aligned}\right.
\end{split}
\end{equation}
The characteristic equation $P(\zeta,\eta):=\det(\eta-L(\zeta))=0$ defines the spectral curve $\mathcal{C}$, a compact algebraic curve of genus $(\charge-1)^2$ which takes the form
\begin{equation}
P(\zeta,\eta):=\eta^\charge+a_1(\zeta)\eta^{\charge-1}+\ldots+a_\charge(\zeta)=0, \quad
\mathrm{deg}\, a_i(\zeta) \leq 2i\label{curve1}.
\end{equation}
The antihermiticity of the Nahm matrices $T_\ast$ means that the curve has a real structure, that is, it is invariant under the
anti-holomorphic involution
\begin{equation}\label{invol}
\mathfrak{J}:\, (\zeta,\eta)\rightarrow
(-1/{\overline\zeta},-{\overline\eta}/{{\overline\zeta}\sp2}).
\end{equation}
Thus Nahm's equations may be associated with a spectral curve endowed with real structure. Hitchin in his seminal paper \cite{Hitchin1983} established a fundamental trichotomy between Euclidean $\mathfrak{su}(2)$ BPS monopoles up to gauge equivalence, Nahm data for the interval $[0,2]$ up to gauge equivalence and such spectral curves. Here the Nahm data ensures that
the monopoles are regular on $\mathbb{R}\sp3$; these conditions on the Nahm data become conditions on the spectral curve.
Variations of this set-up for example allow the description of BPS monopoles for other gauge groups \cite{Hurtubise1989} and Dirac $U(1)$ multimonopoles (see for example \cite{Braden2023b}). Further, integrable systems techniques allow one to associate solutions to Nahm's equations with a linear flow on the Jacobian $\Jac(\mathcal{C})$ \cite{Ercolani1989,Braden2018} and making use of the 
Baker-Akhiezer function one may even reconstruct the solutions of (\ref{eq:ASDequation}) \cite{Braden2021a}.

We shall now focus on Euclidean $\mathfrak{su}(2)$ BPS monopoles. Despite the maturity of the subject the number of explicitly known spectral curves (equivalently monopoles) is few. This is in part because the Nahm data impose transcendental constraints \cite{Braden2021} on the curve $\mathcal{C}$ that we shall presently describe. Those curves and monopoles we do know have been found by using symmetry (see for example \cite{Hitchin1995}). Recently some new examples have been found \cite{Braden2023} and this work extends that.

In particular, we prove the following theorem (Theorem \ref{thm: main Dihedral} in the main text).  
\begin{theorem*}
Any $D_\charge$ rotationally dihedrally symmetric charge-$\charge$ Euclidean $\mathfrak{su}(2)$ BPS monopole is, up to an overall rotation, given by Nahm data gauge equivalent to (\ref{sutcliffeansatz}) obtained from the $\mathfrak{su}(\charge)$ affine Toda equations $A_{\charge-1}\sp{(1)}$ subject to the restriction of the variables (\ref{reductionconds}) arising from a folding procedure. For $k=2l$ and $k=2l-1$ these are respectively the $C_l^{(1)}$ and $A_{2(l-1)}^{(2)}$ affine Toda equations.
\end{theorem*}
In \S\ref{sec: curves symmetry and flows} we describe the action of the $D_\charge$ symmetry on the monopole both via its action on the Nahm data and via its action on the variables of the spectral curve. We also give a brief description of the conditions imposed on the Nahm data / spectral curve, the difficulty present in imposing them, and as such the motivation for finding reductions to more tractable problems by imposing symmetry. In \S\ref{sec: affine toda} we introduce the affine Toda equations including the role of the underlying Lie algebra and the folding procedure. We also review how solutions to the affine Toda equations give solutions to Nahm's equations. \S\ref{sec: dihedral monopoles} applies this to the case of monopoles with $D_\charge$ symmetry, explaining also what the spectral curves of such monopoles are and why Nahm's equations linearise on the Jacobian of a hyperelliptic curve. Finally \S\ref{sec: the theorem} proves our main theorem. We will not describe the relevant solutions here though \S\ref{sec: solutions} will comment on the these. We conclude in \S\ref{sec: discussion} with some discussion.

\section{Curves, Symmetry and Flows}\label{sec: curves symmetry and flows}

Hitchin's description of monopoles may be understood as a dimensional reduction of the twistor correspondence that described (see \cite{Ward1991}) solutions of (\ref{eq:ASDequation}). The reduction of twistor space becomes the
Euclidean mini-twistor space $\mathbb{MT}$, the space of oriented lines in Euclidean 3-space. If the direction of an oriented line is given by $\zeta $, an affine coordinate of $[\zeta_0:\zeta_1]\in \mathbb{P}^1$, and $\eta\in \mathbb{C}$ describes the point in the plane perpendicular to this through which the line passes then we have
$\eta\partial_\zeta\in T\mathbb{P}^1 \cong \mathbb{MT}$. Here
$\mathfrak{J}:\mathbb{MT}\rightarrow \mathbb{MT}$ is the involution that sends a line with orientation to the same line with opposite orientation\footnote{We identify
$\mathbb{MT} = \pbrace{(\bm{u}, \bm{v}) \in S^2 \times \mathbb{R}^3 \, | \, \bm{u} \cdot \bm{v} = 0}$ with $T\mathbb{P}^1$ by
$\zeta :=\zeta_0/\zeta_1 = ({u_1 + i u_2})/({1-u_3})$ and
$\eta = d\zeta({\bf v})= 
({v_1+iv_2+\zeta v_3})/({1-u_3})$ where 
${\bf v}=v_{1}\partial_{u_1}+v_{2}\partial_{u_2}+v_{3}\partial_{u_3}\in T_{\bf u}S^2$. Here $\mathfrak{J}(\bm{u}, \bm{v})=(-\bm{u}, \bm{v})$.
Note that $(u_1,u_2,u_3)=\left(\zeta+\overline{\zeta}, [{ \zeta-\overline{\zeta} }]/{i}, |\zeta|^2 -1\right)/({1+|\zeta|^2})$.
}. The monopole spectral curve $\mathcal{C}$ is then a subset of $T\mathbb{P}^1$, and in fact a compact algebraic curve. Hitchin's construction thus gives us something different from the usual spectral curve of an integrable system: here we also have an interpretation of the extrinsic 
geometry of the curve and the space in which it lies.

\subsection{Actions of symmetries}\label{sec: action of symmetries}

In particular the action of the %orthogonal group $O(3)$ 
Euclidean group $\operatorname{E}(3) = \mathbb{R}^3 \rtimes \Orth(3)$ and symmetries of configurations of monopoles
are reflected in an action on mini-twistor space, the spectral curve and Nahm's equations. For example translations $\bm{x} \in \mathbb{R}^3$ 
act on Nahm data via
$
\bm{x} : T_j \mapsto T_j + i x_j \id_k;
$
in terms of the spectral data this corresponds to the shift
%\begin{equation}\label{eq: transfrom of eta under translation}
$\bm{x} : \eta \mapsto \eta - i\psquare{(x_1 +ix_2) - 2i x_3 \zeta + (x_1 -i x_2)\zeta^2}.
$
%\end{equation}
Such a translation yields a shift in the coefficient $a_1(\zeta)$ of the spectral curve and henceforth we will use this freedom to set $a_1(\zeta)=0$. This gives us a \emph{centred monopole} and the
subgroup of $E(3)$ that fixes this centre is called a 
\emph{point group}; it is a subgroup of $\OO(3)$.

To describe the remaining Euclidean symmetries it is easiest to first describe rotations acting upon $(\zeta,\eta)$; these come from the natural action of $\SO(3)$ on $\mathbb{MT}$. If $A \in \SO(3)$ corresponds to a rotation about $\bm{n} \in S^2$ by angle $\theta$, then 
we may associate to
$
\begin{pmatrix}a+i b&c+id\\-[c-id]&a-ib\end{pmatrix}\in \SU(2)
$
the fractional linear transformation\footnote{As discussed in \S\ref{sec: action of symmetries} the conventions here have been chosen so as to make the action of rotations on the $T\mathbb{P}^1$ coordinates be equal to that on the corresponding Nahm matrices, via the definition of the spectral curve. Our conventions are equivalent to those in \cite[(14), (25)]{Hitchin1995} but differ from those taken in \cite[(8.198), (8.219)]{Manton2004}. The latter conventions give inconsistent actions on the Nahm data, as can be seen by considering a charge-1 monopole.}
\begin{equation}\label{eq: action of rotation on TP1 variables}
A:(\zeta, \eta)\mapsto\left( \frac{ [a+ib]\zeta+[c+id] }{-[c-id]\zeta+[a-ib]}, \frac{\eta}{( -[c-id]\zeta+[a-ib])^2 }
\right)
\end{equation}
where 
$a=\cos(\theta/2)$, $b= n_3\sin(\theta/2)$, $c=n_1\sin(\theta/2)$, $ d=n_2\sin(\theta/2)$. On Nahm data we have the corresponding transformation
$$
A : T_i \mapsto  A_{ij}\psquare{\rho(A) T_j \rho(A)^{-1}} 
$$
where $\rho(A)$ is the image of $A$ in $\SU(\charge)$ in the appropriate representation. (This will be described below once the Nahm data for our
problem is prescribed.)

The two transformations both depends on choices of convention that must be chosen to be compatible. To test compatibility one may consider the $\charge=1$ case where the Nahm matrices are $T_j = i x_j$ for some constant $x_j \in \mathbb{R}$. The spectral curve is then given by 
\[
\eta - \psquare{(ix_1 - x_2) +2x_3 \zeta + (ix_1 + x_2) \zeta^2} = 0. 
\]
We shall use three examples which shall be important later.
\begin{enumerate}
\item A rotation about the $x_1$-axis by angle $\pi$ send $x_2\mapsto -x_2$, $x_3 \mapsto -x_3$, and fixes $x_1$. This means the spectral curve is sent to 
\[
\zeta^2\pbrace{(\eta/\zeta^2) - \psquare{(ix_1 + x_2)(-1/\zeta)^2 +2x_3(-1/ \zeta) + (ix_1 - x_2)}} = 0, 
\]
and we see that the corresponding transform of the $T\mathbb{P}^1$ variables is $(\zeta, \eta) \mapsto (-1/\zeta, \eta/\zeta^2)$. 
\item A similar rotation now about the $x_2$-axis sends the spectral curve to 
\[
-\zeta^2\pbrace{(-\eta/\zeta^2) - \psquare{(ix_1 + x_2)(1/\zeta)^2 +2x_3(1/ \zeta) + (ix_1 - x_2)}} = 0, 
\]
so the transform is $(\zeta, \eta) \mapsto (1/\zeta, -\eta/\zeta^2)$. 
% A rotation about the $x_2$ axis by angle $\pi$ send $x_1\mapsto -x_1$, $x_3 \mapsto -x_3$, and fixes $x_2$. This means the spectral curve is sent to 
% \[
% -\zeta^2\pbrace{(-\eta/\zeta^2) - \psquare{(ix_1 + x_2)(1/\zeta)^2 +2x_3(1/ \zeta) + (ix_1 - x_2)}} = 0, 
% \]
% and we see that the corresponding transform of the $T\mathbb{P}^1$ variables is $(\zeta, \eta) \mapsto (1/\zeta, -\eta/\zeta^2)$. 
    \item A rotation about the $x_3$-axis by angle $2\pi/n$ sends $x_1 + i x_2 \mapsto \omega (x_1 + i x_2)$ and so the spectral curve to 
    \[
\omega \pbrace{(\omega^{-1}\eta) - \psquare{(ix_1 - x_2) +2x_3 (\omega^{-1} \zeta) + (ix_1 + x_2) (\omega^{-1}\zeta)^2}} = 0 \Rightarrow (\zeta, \eta) \mapsto (\omega \zeta, \omega \eta).
\]
%     A rotation about the $x_3$ axis by angle $2\pi/n$ sends $x_1 + i x_2 \mapsto \omega (x_1 + i x_2)$ where $\omega = \exp(2\pi i /n)$ and fixes $x_3$. This means the spectral curve is sent to 
%     \[
% \omega \pbrace{(\omega^{-1}\eta) - \psquare{(ix_1 - x_2) +2x_3 (\omega^{-1} \zeta) + (ix_1 + x_2) (\omega^{-1}\zeta)^2}} = 0, 
% \]
% and we see that the corresponding transform on $T\mathbb{P}^1$ variables is $(\zeta, \eta) \mapsto (\omega \zeta, \omega \eta)$. 
\end{enumerate}

We note that although the full M\"obius group will act on $\zeta$ (and consequently $\eta$) and the curve $\mathcal{C}$ and its transform have the same period matrices, it is only if we restrict to $\SU(2)$ that the real structure is preserved. It remains to describe the action of the elements $A\in \OO(3)\setminus\SO(3)$. These are 
antiholomorphic and follow from\footnote{Note that $-\id$ acts on $\mathbb{MT}$ as $(\bm{u}, \bm{v})\mapsto (-\bm{u}, -\bm{v})$.
The map $-\id$ is called \emph{inversion} in \cite{Houghton1996b}, whereas the reflection $ \diag(1, 1, -1)$ corresponding to the map $(\zeta, \eta) \mapsto (1/\bar{\zeta}, -\bar{\eta}/\bar{\zeta}^2)$
is called inversion in \cite{Hitchin1995}.
}
\[
-\id : (\zeta, \eta) \mapsto (-1/\bar{\zeta}, \bar{\eta}/\bar{\zeta^2}).
\]
By composing the transformation for $-\id$ with a rotation by $\pi$ about the $x_1$-axis we see, for example, that the reflection $\sigma = \diag(-1, 1, 1)$ corresponds to $\sigma : (\zeta, \eta) \mapsto (\bar{\zeta}, \bar{\eta})$. 

Because the action of $E(3)$ on the spectral curve variables is nontrivial, or equivalently because the action on the Nahm matrices is not just by conjugation, these transforms do not give gauge equivalent monopoles. 

\subsection{The Hitchin and Nahm conditions}

It remains to describe the Nahm data for Euclidean $\mathfrak{su}(2)$ BPS monopoles of charge $\charge$, first, following Nahm, for the system (\ref{eq:NahmT0}) and then Hitchin's conditions on the monopole
spectral curve equivalent to the these. They are\footnote{We use the suffix convention $i \in \{1, 2, 3\}$ and $a \in \{1, \ldots, 4\}$.}:
\begin{itemize}
\item[\bf{N1}] Nahm's equations (\ref{eq:NahmT0}) are satisfied,

\item[\bf{N2}] $T_i(s)$ is regular for $s\in(0,2)$ and has simple
poles at $s=0$ and $s=2$, the residues of which form an irreducible
$\charge$-dimensional representation of $\mathfrak{su}(2)$,

\item[\bf{N3}]
$\displaystyle{T_a(s)=-T_a^\dagger(s),\quad T_a(s)=T_a^T(2-s).}$\footnote{Given Nahm matrices satsifying all the Nahm conditions except for $T_a(s) = T_a^T(2-s)$, one can always find a gauge in which this final condition holds \cite{Hitchin1995}.}
%\end{itemize}

%\begin{itemize}

\item[\bf{H0}] $\mathcal{C}$ has no multiple components, i.e. each irreducible component has multiplicity 1.

\item[\bf{H1}] Reality conditions
$\displaystyle{
a_r(\zeta)=(-1)^r\zeta^{2r}\overline{a_r(-1/{\overline{\zeta}}\,)}},$

\item[\bf{H2}] Let  $L^{s}$ denote the holomorphic line bundle on
    $T\PP$ defined by the transition function
    $g_{01}=\rm{exp}(-s\eta/\zeta)$
and let
 $L^{s}(m)\equiv
L^{s}\otimes\pi\sp*\mathcal{O}(m)$ be similarly defined in
terms of the transition function
$g_{01}=\zeta^m\exp{(-s\eta/\zeta)}$. Then $L^2$ is trivial on
${\mathcal{C}}$ and $L\sp1(\charge-1)$ is real.

\item[\bf{H3}] $H^0({\mathcal{C}},L^{s}(\charge-2))=0$ for
    $s\in(0,2)$.
\end{itemize}
We have already encountered N1, N3(a) and the reality of the spectral curve H1; the remaining conditions, which encode the regularity of the
gauge fields, restrict the curve. The poles of the Nahm data at $s=0,2$ corresponds to a flow on the Jacobian that begins and ends on the $\Theta$-divisor. This identification follows by observing that
$\deg L^s(\charge-2) = \charge(\charge-2) = g(\mathcal{C})-1$, and so we may view $L^s(\charge-2)$ as a straight line curve in $W_{g-1}$ (the image of effective degree-$(g-1)$ divisors in the Jacobian via the Abel-Jacobi map) of period 2 in the direction $\psquare{\eta/\zeta} \in T_{L^s(k-2)}\Pic^{g-1}(\mathcal{C}) \cong H^1(\mathcal{C}, \mathcal{O}_\mathcal{C})$, not intersecting the theta divisor for any $s \in (0,2)$. (Changing the basepoint of the AJ map which maps the degree-$(g-1)$ divisor into $W_{g-1} \subset \Jac(\mathcal{C})$ just corresponds to a translation, so this statement is basepoint independent.) Via the Riemann vanishing theorem \cite[p298]{Farkas1992} only points
on the $\Theta$-divisor have $H^0({\mathcal{C}},L^{s}(n-2))\ne0$. Now 
$[\eta/\zeta]$ corresponds to the set of Laurent tails $\pbrace{r_{0_j}}$, $r_{0_j} = \frac{\eta_j(0)}{\zeta}$, 
where  $(\zeta, \eta_j(\zeta))$ are the preimages of $\zeta$ in $\mathcal{C}$ and we set $0_j = (0, \eta_j(0))$, $\infty_j=\mathfrak{J}(0_j)$.  
Under Serre duality this gives the linear map on holomorphic differentials
%\begin{equation*}\label{eq: sum of resdiues around 0}
$\omega \mapsto \sum_j \Res_{0_j} \pround{\frac{\eta}{\zeta}\omega}$ .
%\end{equation*}
We get the coordinates of a vector $\bm{U}$ in the Jacobian viewed as $\mathbb{C}^g/\Lambda$ by fixing a canonical homology basis $\pbrace{\mathfrak{a}_j, \mathfrak{b}_j}$ and basis of $\mathfrak{a}$-canonically normalised differentials $\pbrace{\nu_j}$ (i.e. $\int_{\mathfrak{a}_j} \nu_k = \delta_{jk}$, $\int_{\mathfrak{b}_j} \nu_k = \tau_{jk}$, where $\tau$ is the Riemann matrix). 
%The $j$th entry of the vector in $\mathbb{C}^g$ corresponding to a linear map $f \in H^0(\mathcal{C}, K_\mathcal{C})^\ast$ is now $f(\nu_j)$. 
Then the $j$-th entry of the vector in $\mathbb{C}^g$ corresponding to $[\eta/\zeta]$ is 
\[
U_j = \sum_{l} \Res_{0_l}\pround{\frac{\eta}{\zeta} \nu_j}
= \frac{1}{2 \pi i} \oint_{\mathfrak{b}_j} \gamma_0 = \frac{1}{2 \pi i} \oint_{\mathfrak{b}_j} \gamma_\infty.
\] 
The final expressions come from using the reciprocity law for differentials and by choosing a differential of the second kind
$\gamma_0$ (respectively $\gamma_\infty$) such that $ \gamma_0 \sim d(\eta/\zeta)$ around $0_j$ ($\gamma_\infty \sim d(-\eta/\zeta)$ around $\infty_j$) and $\gamma_0$ ($\gamma_\infty$) is holomorphic everywhere else. 
These final expressions were introduced by Ercolani and Sinha \cite{Ercolani1989} and we call $\bm{U}$ the \emph{Ercolani-Sinha} vector.
\begin{lemma}[Ercolani-Sinha Constraints \cite{Ercolani1989,Houghton2000,Braden2010}] The following are equivalent:
\begin{enumerate}[(i)]
 \item $L\sp2$ is trivial on $\mathcal{C}$.

\item There exists a 1-cycle
$\mathfrak{es}=\boldsymbol{n}\cdot{\mathfrak{a}}+
\boldsymbol{m}\cdot{\mathfrak{b}}$ for $(\bm{n}, \bm{m}) \in \mathbb{Z}^{2g}$ such that for every holomorphic
differential $\Omega=\left({\beta_0\eta^{\charge-2}+\beta_1(\zeta)\eta^{\charge
-3}+\ldots+\beta_{\charge-2}(\zeta)}\right) d\zeta/({{\partial {P}}/{\partial
\eta}})$,
\begin{equation}
\oint\limits_{\mathfrak{es}}\Omega=-2\beta_0,\label{HMREScond}
\end{equation}
\item
$\displaystyle{
2\boldsymbol{U}\in \Lambda\Longleftrightarrow
\boldsymbol{U}=\frac{1}{2\pi\imath}\left(\oint_{\mathfrak{b}_1}\gamma_{\infty},
\ldots,\oint_{\mathfrak{b}_g}\gamma_{\infty}\right)\sp{T}= \frac12
\boldsymbol{n}+\frac12\tau\boldsymbol{m} .}$

\end{enumerate}
\end{lemma}
Thus for a monopole spectral curve  we require that  $\boldsymbol{U}$ is a half-period.
Further, from {H3} the vector $\boldsymbol{U}$ should be {\it primitive}, i.e. $s\boldsymbol{U}\not\in\Lambda$ for $s\in(0,2)$. We call $\mathfrak{es}$ the \emph{Ercolani-Sinha cycle} and if it exists for a curve $\mathcal{C}$ then it is unique. The Ercolani-Sinha Constraints place
$g$ transcendental constraints on the curve $\mathcal{C}$ \cite{Braden2021}.

We have then that the spectral curve of a Euclidean $\mathfrak{su}(2)$ BPS monopole comes equipped with a particular half-period. From $\mathfrak{J}_\ast 
\mathfrak{es}=-\mathfrak{es}$ or $\mathfrak{J}_\ast 
\boldsymbol{U}=-\boldsymbol{U}$ then $2\boldsymbol{U}$ may be identified with the imaginary lattice point of Hitchin \cite[p164]{Hitchin1983}.
Although we do not know how to impose the Ercolani-Sinha Constraints, given a curve satisfying them it is algorithmic to use numerical methods to compute the integers $(\bm{n}, \bm{m})$.

\subsection{Symmetry reduction and hyperelliptic curves}

One may ask about the action of $\Aut(\mathcal{C})$ on $\mathfrak{es}$, or equivalently the action on $\bm{U}$. In \cite{Houghton2000} it was shown that $\mathfrak{es}$ is invariant under the $A_4$ tetrahedral subgroup of the $S_4 \times C_3$ full automorphism group of the tetrahedral 3-monopole spectral curve, and more generally it was shown in \cite{Houghton2000, Braden2011} that $\mathfrak{es}$ is invariant under the action of any rotation. One may indeed verify that for
the tetrahedral 3-monopole $\mathfrak{es}$ is invariant under the $S_4$ subgroup of the automorphism group, but not under the $C_3$ action $(\zeta,\eta)\mapsto (\zeta,e^{2\pi i/3} \eta)$ which does not correspond to any transformation in $E(3)$.
In good cases the invariance of the spectral curve $\mathcal{C}$ and the Ercolani-Sinha vector under $G\le\Aut(\mathcal{C})$
means that the flow on $\mathcal{C}$ comes from the pull back of a flow
on $\mathcal{C}/G$: this allows simplifications (see \cite{Braden2011,Braden2011a}).

We have now discussed spectral curves, some of their symmetries, and the flows that are needed to describe BPS monopoles, but
before proceeding it is helpful to note
a less obvious feature of the spectral curve $\mathcal{C}$:
 it is non-hyperelliptic for charge-$\charge\ge3$, a fact which follows from classical algebraic geometry \cite[\S2]{Braden2023}.
% To see this non-hyperellipticity consider the 
% compactification of $T\mathbb{P}^1$. Two compactifications of this are common, either by inclusion in the (singular) weighted projective space 
% $
% \mathbb{P}\sp{1,1,2}=
% \{(\zeta_0,\zeta_1, \eta)\in\mathbb{C}\sp3\setminus\{0\} \,|\, (\zeta_0,\zeta_1,\eta)\sim(\lambda\zeta_0,\lambda\zeta_1,\lambda^2\eta), \ \lambda\in\mathbb{C}\sp\ast \},
% $
% or (as by Hitchin) in the Hirzebruch surface $\mathbb{F}_2$. We adopt the former view and note that the singular point $[0:0:1]$ does not lie in  $T\mathbb{P}^1$ and hence on $\mathcal{C}$.
% Next, via the Veronese embedding, we have $\iota:\mathbb{P}\sp{1,1,2}
% \hookrightarrow\mathbb{P}\sp{3}$, $\iota([\zeta_0:\zeta_1: \eta])=
% [\zeta_0^2:\zeta_0\zeta_1:\zeta_1^2: \eta]$. Under this a homogeneous polynomial of degree $2r$ becomes a homogeneous polynomial of degree $r$ in the new coordinates and $\mathbb{P}\sp{1,1,2}$ becomes a cone, a quadric, over the cone point $\iota([0:0:1])$ (see, for example, \cite[\S 8.2.11]{Vakil2010}). Thus a monopole spectral curve may be viewed as the complete intersection of a quadric cone and a degree-$\charge$ hypersurface in $\mathbb{P}\sp{3}$. This is known to be a curve of genus $(\charge-1)^2$ (see for example \cite[Exercise V.2.9]{Hartshorne1977}) which is non-hyperelliptic for $\charge\ge3$ (\cite[Exercise IV.5.1]{Hartshorne1977}). 
One immediate consequence is that generic monopoles cannot be described by familiar integrable systems whose spectral curves are hyperelliptic. In what follows we shall show however that some monopoles constrained by symmetry are so expressible.

\section{The Affine Toda Equations}\label{sec: affine toda}
In what follows we shall refer to the affine Toda equations and this section recalls these and a number of their properties. We begin by describing the untwisted algebras and then discuss their reduction by folding. Folding is a procedure for constructing non-simply laced Lie algebras from simply laced algebras by means of a
Dynkin diagram automorphism $\tau$, for both finite and affine algebras. Such
an automorphism is an outer automorphism of the algebra and the $\tau$-invariant subalgebra is then the desired algebra. The untwisted finite and affine algebras are most easily seen, arising from a diagram 
automorphism of the finite algebra; describing the twisted affine algebras is somewhat more involved. For the purposes of this paper
we may circumvent these constructions and just use the resulting Cartan matrices together with an observation of Olive and Turok
\cite{Olive1983}:
as the Toda equations are based on the Cartan matrices and folding arises by a diagram automorphism, folding also
extends to these equations and we obtain the Toda equations of the folded algebra.

\subsection{The equations for \secmath{\mathfrak{g}\sp{(1)}}}
Let $\mathfrak{g}$ be a simple Lie algebra (over $\mathbb{C}$) of rank $r$ and $\mathfrak{h}\subset \mathfrak{g}$ 
a fixed Cartan subalgebra. We have the root space decomposition $\mathfrak{g}=\mathfrak{h}
\bigoplus\left(\oplus_{\alpha\in\Phi}\mathfrak{g}_\alpha\right)$
where $\Phi\subset\mathfrak{h}\sp\ast$ denotes the set of roots for the pair $(\mathfrak{g},\mathfrak{h})$. Endow $\mathfrak{h}\sp\ast$ with
the  inner product $(\  ,  \ )$ and let $W$ be the associated Weyl group. By averaging we may always take  $(\  ,  \ )$ to be Weyl-invariant. If $\kappa$ is the Killing form on $\mathfrak{g}$
the isomorphism $\nu:
{\mathfrak{h}}\rightarrow{\mathfrak{h}}\sp\ast$ defined by
$\kappa(h_1 ,  h_2 )=\nu(h_1)(h_2)$ (for $h_{1,2}\in{\mathfrak{h}}$)
enables the identification $\kappa(h_1 ,  h_2 )=(\nu(h_1),\nu(h_2))/c$
where $c$ is a constant to be defined below.
Choose a set of simple roots $\Delta:=\{\alpha_1,\ldots,\alpha_r\}\subset\Phi$.
To each $\alpha\in\Phi$ set
 $\epsilon_\alpha:=2/{(\alpha,\alpha)}$, $\alpha\sp\vee :=\epsilon_\alpha \alpha:={2\alpha}/{(\alpha,\alpha)}$, the dual root.
  We write $\epsilon_i:={2}/{(\alpha_i ,\alpha_i)}$
 for $\alpha_i\in\Delta$. The Cartan matrix is\footnote{Conventions here
 differ between authors: we follow \cite[p71]{Carter2005} which is the transpose
 of \cite[X Lemma 3.3]{Helgason1978}.} $K:=(a_{ij})$ with $ a_{ij}:=(\alpha_i\sp\vee,\alpha_j)$.
The corresponding Chevalley basis of $\mathfrak{g}$ consists of
 $\{H_i:=H_{\alpha_i}\}_{\alpha\in\Delta}$, the basis of $\mathfrak{h}$, and $\{E_\beta\}_{\beta\in\Phi}$, 
then satisfies
\begin{align*} 
[H_\alpha, E_\beta]&=(\alpha\sp\vee,\beta)\, E_\beta,\quad
[E_\alpha,E_\beta]&=
\begin{cases} H_\alpha
&\text{if }\alpha+\beta=0,\\
N_{\alpha,\beta}E_{\alpha+\beta}&\text{if }\alpha+\beta\in\Phi,\\
0&\text{if }\alpha+\beta\ne0\text{ and }\alpha+\beta\not\in\Phi.
\end{cases}
\end{align*}
Here (for all $\alpha\in\Phi$) we have
$H_\alpha=\sum_{i=1}\sp{r}c_i H_i$
where the $c_i$ are defined by $\alpha\sp\vee=\sum_{i=1}\sp{r}c_i \alpha_i\sp\vee$; the $c_i\in\mathbb{Z}$ as $\{\alpha_i\sp\vee\}$ is a basis for $\Phi\sp\vee:=\{\alpha\sp\vee\,|\, \alpha\in\Phi\}$.
We will not need to further specify $N_{\alpha,\beta}$ for what follows.
One finds $\nu(H_\alpha)=c\, \alpha\sp\vee$.
Let
$\Theta=\sum_{\alpha\in\Delta} {n_\alpha}\, \alpha$ be the highest root of $\Phi$ and set 
$\bar\Delta=\Delta\cup\{-\Theta\}$ and also $n_{-\Theta}=1$. The highest root is always a long root and we normalize so that $(\Theta,\Theta)=2$. Further denoting $\alpha_0=-\Theta$ we have the extended Cartan matrix $\overline{K}_{ij}$ for $i,j\in 0,\ldots,r$ which has right null-vector
$n_i$ ($\equiv n_{\alpha_i}$) and left null-vector $m_i=n_i/\epsilon_i$ ($\equiv m_{\alpha_i}$). 
The Coxeter number for $\mathfrak{g}$ is then $h=\sum_{\alpha\in \bar\Delta}n_\alpha$ with the dual Coxeter number being
$g=\sum_{\alpha\in \bar\Delta}m_\alpha$; the constant $c$ introduced earlier is then $c=2g$ for all root systems apart from $A_{2l}\sp{(2)}$
which has three root lengths.
The Cartan matrix for the the untwisted affine Kac-Moody algebra $\mathfrak{g}\sp{(1)}$ is then $\overline{K}$.

The key property that we now exploit\footnote{\cite{Bogoyavlensky1976} calls such an admissible root system.} is that the difference of any two vectors in $\Delta$ or $\bar\Delta$ (or the roots associated with the twisted constructions below) is not a root. In particular \cite[(2.19)]{Olive1993}, upon setting
$$E=\sum_{\alpha\in \bar\Delta} \sqrt{m_\alpha}\, E_\alpha,\quad
E\sp\dagger=\sum_{\alpha\in \bar\Delta} \sqrt{m_\alpha}\, E_{-\alpha}.
$$
we have (noting that  $\theta\sp\vee=\theta= \sum_{\alpha\in\Delta}  {n_\alpha}\alpha=\sum_{\alpha\in\Delta}  {m_\alpha}\alpha\sp\vee$) the vanishing of
$$[E,E\sp\dagger]=\sum_{\alpha\in\Delta}  {m_\alpha}\,[E_\alpha, E_{-\alpha}]+[ E_{-\Theta},E_\Theta]
=\sum_{\alpha\in\Delta}  {m_\alpha}\,H_\alpha -H_\Theta
=0.$$
We now give an ansatz for a Lax pair that encodes the affine Toda equations for $\mathfrak{g}\sp{(1)}$ in a form that makes connection with (\ref{nahmlax}). 
\begin{lemma} Nahm's equations for the Lax matrices (\ref{nahmlax})  become the affine Toda equations with the following ansatz:
let $T_4=0$ (which  means $\alpha=\alpha\sp\dagger$) and taking\footnote{This should be viewed as a suggestive notation; when we fix a representation $\rho$ of the algebra below we will have $\rho(E_\alpha)\sp\dagger =\rho(E_{-\alpha})$.} $\phi=\phi\sp\dagger\in\mathfrak{h}$, $E_\alpha\sp\dagger =E_{-\alpha}$, set
\begin{equation}\label{nahmansatz}
 \beta=T_1+i T_2=e^{\phi/2} E e^{-\phi/2},\
\beta\sp\dagger=-T_1+i T_2=e^{-\phi/2} E\sp\dagger e^{\phi/2},\
\alpha+\alpha\sp\dagger=2i T_3 ={\dot \phi}.
\end{equation}
\end{lemma}
\begin{proof}
The reality of $q$ and $p=\dot q$ is equivalent to the hermiticity requirements that
$T_a\sp\dagger=-T_a$.
Then Nahm's equations are (see also \cite[(3.5)]{Braden2018})
$$
\ddot \phi=2i\dot{T}_3=
[\beta, \beta\sp\dagger]=e^{-\phi/2}[ e^{\phi}E  e^{-\phi} ,E\sp\dagger ]   e^{\phi/2}
$$
which are the Toda equations. Setting $h=e^{\phi}$ we have
\begin{align*} 
\ddot \phi&=[ e^{\phi}E  e^{-\phi} ,E\sp\dagger ] 
\Longleftrightarrow
\frac{d}{ds} \left( \dot{h} h\sp{-1} \right)= \left[ h E h\sp{-1}, E\sp\dagger\right].
\end{align*}
Using
$e^{\phi}E_\alpha  e^{-\phi}=\Ad_{e^{\phi}}E_\alpha 
=e^{\alpha(\phi)}E_\alpha$ and our earlier remarks about $E$ this becomes
$$\ddot \phi=\sum_{i=1}\sp{r}m_i\left[ e\sp{\alpha_i(\phi)}-
 e\sp{-\Theta(\phi)}\right]H_i =
 \sum_{\alpha\in \bar\Delta}m_\alpha \, e\sp{\alpha(\phi )}H_\alpha$$
which may be cast into a (perhaps) more familiar form under the isomorphism $\nu$. With $\varphi =\nu(\phi)/c$ and noting then that $\alpha(\phi)=(\alpha,\varphi )$ we have
\begin{equation}\label{generaltoda}
  \ddot \varphi  =  \sum_{\alpha\in \Delta}n_\alpha \alpha\, e\sp{(\alpha,\varphi )}
  -\Theta\, e\sp{-(\Theta,\varphi )}
  =\sum_{\alpha\in \bar\Delta}n_\alpha \alpha\, e\sp{(\alpha,\varphi )}.
\end{equation}
If we introduce the fundamental weights $\{\lambda_j\}_{j=1}\sp{r}$ for
which $(\lambda_j,\alpha_j\sp\vee)=\delta_{ij}$ and set $\varphi =\sum_{i=1}\sp{r} \varphi _{i}\epsilon_i\lambda_i$ calling $\varphi _{0}=-(\Theta,\varphi )=-\sum_{i=0}\sp{r}n_i \varphi _i$ we have (for $i\in\{0,\ldots,r\}$)
\begin{equation}\label{eqntoda}
    {\ddot \varphi }_i = \sum_{j=0}\sp{r} {\overline K}_{ij}\sp{T}\, m_j\, e\sp{\varphi _j}.
\end{equation}
Clearly the variable $\varphi _0$ is redundant but makes the equations more symmetric. Further $\sum_{i=0}\sp{r}m_i p_i=0$ (where $p_i=\dot \varphi _i$) for consistency; for the $A_r\sp{(1)}$ case this is usually interpreted as the momentum of the centre of mass. The associated Hamiltonian is $H=(p,p)/2 +\sum_{\alpha\in \bar\Delta}n_\alpha\, e\sp{(\alpha,\varphi )}$.
\end{proof}

Thus far we have only used the properties of the root system of the finite dimensional Lie algebra $\mathfrak{g}$. We will not need the full construction of the affine algebra $\mathfrak{g}\sp{(1)}$ but note that the standard construction(see \cite{Kac1990,Carter2005})
proceeds via the loop algebra $\mathbb{C}[w,w\sp{-1}]\otimes_\mathbb{C}\mathfrak{g} $ and has generators $e_i=1\otimes E_{\alpha_i}$, $f_i=1\otimes E_{-\alpha_i}$,
$h_i=1\otimes H_{i}$ ($i=1,\ldots,r$) together with
$e_0=w\otimes E_{-\Theta}$, $f_0=w\sp{-1}\otimes E_{\Theta}$ along with $h_0$ and some other generators whose exact form we will not need to specify. If $D(\mathfrak{g})$ (respectively $D(\mathfrak{g}\sp{(1)})$) denotes the Dynkin diagram of $\mathfrak{g}$ ($\mathfrak{g}\sp{(1)}$) with diagram automorphisms
$\Aut(D(\mathfrak{g}))$ ($\Aut(D(\mathfrak{g}\sp{(1)}))$)
we have
$$\Aut(D(\mathfrak{g}))\le \Aut(D(\mathfrak{g}\sp{(1)}))\le \Aut(\Phi(\mathfrak{g}))$$
as the roots of $D(\mathfrak{g})$ are a subset of those of $D(\mathfrak{g}\sp{(1)})$ which are themselves a subset of $\Phi(\mathfrak{g})$. Let $\tau\in\Aut(D(\mathfrak{g}))$ have order
$N$; this fixes
$\Theta$. We may extend $\tau$ to a \emph{twisted automorphism} $\hat{\tau}$ of 
$\mathfrak{g}\sp{(1)}$ with an action on the generators above given by
$\hat{\tau}(w\sp{i}\otimes x)= \delta\sp{-i} w\sp{i}\otimes \tau( x)$, where $\delta$ is a primitive $N$-th root of unity and $x\in\mathfrak{g}$. Then we have the the twisted affine algebra $\mathfrak{g}\sp{(N)}:=(\mathfrak{g}\sp{(1)})\sp{\hat\tau}$ as the
fixed algebra;  there is a range of notation for these.
For our examples below we note  $\Aut(D(A_{r}))\cong C_2$ and 
$\Aut(D(A_{r}\sp{(1)}))\cong D_{r+1}$
for $r\ge2$;  for $r=1$ we have $\Aut(D(A_{1}))$ trivial and 
$\Aut(D(A_{1}\sp{(1)}))\cong C_2$. For $A_r$ ($r\ge2$) we have  $X\mapsto -X\sp{T}$ is an outer automorphism \cite[IX Theorem 5]{Jacobson1979}.

\subsection{$A_{\charge-1}\sp{(1)}$}
Here $n_i=m_i=1$ and (\ref{nahmansatz}) yields
Sutcliffe's ansatz \cite{Sutcliffe1996a}. Expressing the roots in terms of a Euclidean basis 
$\alpha_i=e_i-e_{i+1}$ ($i=1,\ldots,\charge-1$), $|e_i|^2=1$, we have
(with indices being taken mod $\charge$ so $q_\charge=q_0$) in the
$\charge$-dimensional representation
\begin{align}\label{sutcliffeansatz}
T_1+iT_2&=\begin{pmatrix} 0&e\sp{(q_1-q_2)/2}&0&\ldots&0\\
0&0&e\sp{(q_2-q_3)/2}&\ldots&0\\
\vdots&&&\ddots&\vdots\\
0&0&0&\ldots&e\sp{(q_{\charge-1}-q_\charge)/2}\\
w\,e\sp{(q_\charge-q_1)/2}&0&0&\ldots&0
\end{pmatrix} \nonumber\\
T_1-iT_2&=-\begin{pmatrix}0&0&\ldots&0&e\sp{(q_\charge-q_1)/2}/w\\
e\sp{(q_1-q_2)/2}&0&\ldots&0&0\\
0&e\sp{(q_2-q_3)/2}&\ldots&0&0\\
\vdots&&\ddots&&\vdots\\
0&0&\ldots&e\sp{(q_{\charge-1}-q_\charge)/2}&0
\end{pmatrix}\\
T_3&=-\frac{i}{2}\begin{pmatrix} p_1&0&\ldots&0\\
0&p_2&\ldots&0\\
\vdots&&\ddots&\vdots\\
0&0&\ldots&p_\charge
\end{pmatrix}\nonumber
\end{align}
where $p_i$, $q_i$ are real and we see $\sum_{i=1}\sp{\charge}p_i=0$; for this algebra the
center of mass coordinate, corresponding to the centre of the monopole,
may be reinstated by working with the Lie algebra $\mathfrak{gl}(\charge,\mathbb{C})$.
Setting $w=1$ the equations of motion are $\dot p_i=e\sp{q_{i}- q_{i+1}}-e\sp{q_{i-1}-q_i}$
which follow from the affine Toda Hamiltonian
\begin{equation*}\label{affToda}
H=\frac12\left(p_1\sp2+\ldots+p_\charge\sp2\right)-\left[e\sp{q_1-
q_2}+e\sp{q_2- q_3}+\ldots+e\sp{q_{\charge}- q_1} \right].
\end{equation*}
In terms of the Lax matrix
$$\frac12\Tr L^2(\zeta)=\frac12\Tr\left[\beta -(\alpha+\alpha\sp\dagger)\zeta-\beta\sp\dagger \zeta^2\right]^2
=\zeta^2\Tr\left(\frac12 {\dot\phi}^2 -e^{\phi}E  e^{-\phi} E\sp\dagger \right):=\zeta^2 H
$$
upon using $0=\Tr E^2=\Tr \dot \phi(\beta-\beta\sp\dagger)$.
This Hamiltonian is not bounded below  corresponding to a potential of the wrong sign which is necessary as the monopole boundary conditions require $T_i\sim \rho_i/s $ as $s\sim 0$, where $\{\rho_i\}$ is the irreducible $\charge$-dimensional representation. We also define the Flaschka variables $a_i =  e\sp{(q_i-q_{i+1})/2}=e\sp{\varphi_i/2}$ and $b_i=  p_i$ (with indices
taken mod $\charge$). Then
\begin{equation}\label{flaschkaeom}
\dot a_i= \frac12 a_i(b_i-b_{i+1}),\quad
\dot b_i=a_{i}^2 -a_{i-1}^2 
\end{equation}
again following from $
H=\frac12 \sum_{i=1}\sp{\charge} b_i^2-\sum_{i=0}\sp{\charge-1} a_i^2$.
In terms of Flashka variables we have upon reinstating $w$ that
\begin{align}\label{flaschkalax}
{ \tilde{L}}&=\begin{pmatrix} -b_1\zeta&a_1&0&\ldots&0&- a_0 \zeta^2/w\\
-a_1\zeta^2&-b_2\zeta&a_2&\ldots&0&0\\
\vdots&&&\ddots&\vdots\\
0&0&0&\ldots&-b_{n-1}\zeta&a_{n-1}\\
w a_0&0&0&\ldots&-a_{n-1}\zeta^2&-b_n\zeta
\end{pmatrix} \\
\det(\eta -\tilde L(\zeta)) &=\eta\sp\charge +\left(\sum_i b_i\right)\eta^{\charge-1}\zeta+\ldots + (-1)^{\charge-1} \left(\prod_{i=0}\sp{\charge -1}a_i\right)\left(\zeta\sp{2\charge}/w +(-1)\sp{\charge} w\right), \label{flaschkacurve}
\end{align}
and we identify $a_1=0$.

\subsection{$C_l\sp{(1)}$} ($l\ge2$) This is given by the folding
{$A_{2l-1}\sp{(1)}\rightarrow C_{l}\sp{(1)}$}. Given a Euclidean basis 
$\{e_i\}$ with $|e_i|^2=1/2$ we may take simple roots
$\alpha_i=e_i-e_{i+1}$ ($i=1,\ldots,l-1$), $\alpha_l=
2e_l$. Then $2=n_1=\ldots=n_{l-1}$, $n_l=1$ and
$\Theta=2e_1$. We have $1=m_0=m_1=\ldots=m_{l}$. This has Cartan matrix

\begin{equation}\label{cn1diagram}
\overline{K}=\quad \bordermatrix{& 0 & 1 & 2 & .  & .& . & l-2 & l-1 & l \cr
0 & 2&-1 \cr
1 & -2&2&-1 \cr 
2 & &-1&2&. \cr
. &&&.&.&.& \cr
 . & &&&.&.&.& \cr 
 .& &&&&.&.&.& \cr 
 l - 2 & &&&&&.&2&-1& \cr 
l-1 & &&&&&&-1&2&-2 \cr
 l& &&&&&&&-1&2}
% \begin{pmatrix}
%     2&-1\\ -2&2&-1\\ &-1&2&. \\ &&.&.&.& \\
%   &&&.&.&.& \\  &&&&.&.&.& \\  
%   &&&&&.&2&-1& \\ &&&&&&-1&2&-2 \\ &&&&&&&-1&2
% \end{pmatrix}
\end{equation}

\begin{center}
\begin{figure}
\begin{tikzcd}
&&\overset{0}{\circ}\arrow[dll, no head]\arrow[drr, no head]&&
\\
\overset{1}{\circ} \arrow[r,no head]\arrow[bend right=60,leftrightarrow, rrrr] & \overset{2}{\circ} \arrow[r,no head]
\arrow[bend right=60,leftrightarrow, rr]&\circ\cdots\cdots\circ \arrow[r,no head]& \overset{2l-2}{\circ} \arrow[r,no head]& \overset{2l-1}{\circ} \\
\end{tikzcd}
\begin{tikzcd}
\overset{0}{\circ}\arrow[Rightarrow, r]& \overset{1}{\circ} \arrow[r,no head]&
\circ\cdots\cdots\circ& \overset{l-1}{\circ} \arrow[l,no head]\arrow[Leftarrow,r]
& \overset{l}{\circ}
\end{tikzcd}
\caption{The folding $A_{2l}^{(1)} \to C_l^{(1)}$. The vertices of the folded diagram are the orbits of the
graph automorphism.
The edges of the folded diagram are single unless two incident edges map to the same edge in which case these
multiple edges are retained and point from the (fixed) incident
vertex.
}
\end{figure}
\end{center}

\subsection{$A_{2l}\sp{(2)}$} This twisted affine algebra is rather exceptional having three root lengths (see \cite{Carter2005}). If we
choose $\alpha_0$ to be the longest root we have $n_0=1$, $n_1=\cdots=n_{l-1}=n_l=2$. Further taking $(\Theta,\Theta)=2$ then
$(\alpha_i,\alpha_i)=1$ ($i=1,\ldots{l-1}$), $(\alpha_l,\alpha_l)=1/2$
and then $m_0=m_1=\cdots-m_{l-1}=1$ \emph{but} $m_l=1/2$ is no longer an integer.

\begin{center}
\begin{figure}
\begin{tikzcd}
\overset{0}{\circ}\arrow[Rightarrow, r]& \overset{1}{\circ} \arrow[r,no head]&
\circ\cdots\cdots\circ& \overset{l-1}{\circ} \arrow[l,no head]\arrow[Rightarrow,r]
& \overset{l}{\circ}
\end{tikzcd}
\caption{Labelling of the simple roots of $A_{2l}\sp{(2)}$}
\end{figure}
\end{center}

The Cartan matrix for $A_{2l}\sp{(2)}$ is
\begin{equation}\label{a2n2diagram}
\overline{K}=\quad \bordermatrix{& 0 & 1 & 2 & .  & .& . & l-2 & l-1 & l \cr
0 & 2&-1 \cr
1 & -2&2&-1 \cr 
2 & &-1&2&. \cr
. &&&.&.&.& \cr
 . & &&&.&.&.& \cr 
 .& &&&&.&.&.& \cr 
l - 2 & &&&&&.&2&-1& \cr 
l-1 & &&&&&&-1&2&-1 \cr
l& &&&&&&&-2&2}
% \begin{pmatrix}
%     2&-1\\ -2&2&-1\\ &-1&2&. \\ &&.&.&.& \\
%   &&&.&.&.& \\  &&&&.&.&.& \\  
%   &&&&&.&2&-1& \\ &&&&&&-1&2&-1 \\ &&&&&&&-2&2
% \end{pmatrix}
\end{equation}
The Lax matrix for $A_{4}\sp{(2)}$ which is representative for more general $l$ is
$$
L(\zeta):=\begin{pmatrix}
b_{1} \zeta  & a_{1} & 0 & 0 & -{a_{0} \zeta^{2}}/{w} 
\\
 -a_{1} \zeta^{2} & b_{2} \zeta  & a_{2}  & 0 & 0 
\\
 0 & -a_{2} \zeta^{2}  & 0 & a_{2}  & 0 
\\
 0 & 0 & -a_{2} \zeta^{2}  & -b_{2} \zeta  & a_{1} 
\\
 a_{0} w  & 0 & 0 & -a_{1} \zeta^{2} & -b_{1} \zeta  
\end{pmatrix};
$$
this yields a spectral curve
\begin{align} \nonumber
0=\det(\eta-L)&=
\eta^{5}+\left(a_{0}^{2}+2 a_{1}^{2}+2 a_{2}^{2}-b_{1}^{2}-b_{2}^{2}\right) \zeta^{2} \eta^{3}
 +\left(2 a_{0}^{2} a_{2}^{2}-a_{0}^{2} b_{2}^{2}+a_{1}^{4}
\right. \\ &\quad \left.  \label{a42curve}
+2 a_{2}^{2} a_{1}^{2}+
 2 b_{1} b_{2} a_{1}^{2}-2 b_{1}^{2} a_{2}^{2}+b_{1}^{2} b_{2}^{2}\right) \zeta^{4} \eta 
+ a_{0} a_{1}^{2} a_{2}^{2}\left( \frac{ \zeta^{10}}{w}- w \right)
\end{align}
where we observe that $a_{0} a_{1}^{2} a_{2}^{2}$ is a Casimir.

\section{Dihedrally Symmetric Monopoles and Affine Toda Equations}\label{sec: dihedral monopoles}
In the previous section we encountered Sutcliffe's ansatz (\ref{sutcliffeansatz}) that reduced Nahm's equations for charge-$\charge$ to the $A_{\charge-1}\sp{(1)}$ affine Toda equations.
The resulting spectral curve had $C_\charge$ cyclic symmetry.
This work was later strengthened to give
\begin{theorem}[{\cite[Theorem 3.1]{Braden2011}}] Any $C_\charge$ cyclically symmetric charge-$\charge$ Euclidean $\mathfrak{su}(2)$ BPS monopole is, up to an overall rotation, given by Nahm data gauge equivalent to (\ref{sutcliffeansatz}) obtained from the $\mathfrak{su}(\charge)$ affine Toda equations $A_{\charge-1}\sp{(1)}$.
\end{theorem}

This theorem has computational significance.
For a fixed subgroup $G \leq \SO(3)$ it was shown in \cite{Hitchin1995, Houghton1996, Houghton1996a, Houghton1996c, Braden2023, DisneyHogg2023} how to use invariant theory in order to construct $G$-invariant Nahm matrices. The resulting
ODE's in this construction are typically rather opaque, but knowing
the existence of the form (\ref{sutcliffeansatz}) it was shown in
\cite{Braden2023} how these ODE's could be regrouped to reproduce
this form.

The theorem does not assert the existence of appropriate solutions (for which a separate argument is required); Sutcliffe's elliptic solution
for $\charge=3$ failed though a solution in terms of functions on a genus-2 curve was given in \cite{Braden2011a}. Recently new elliptic solutions within this ansatz were found \cite{Braden2023}. These new solutions have an enhanced dihedral symmetry and satisfy a reduction of the affine Toda equations.
In this section we will discuss dihedrally symmetric monopoles
and their connection with (particular) affine Toda Equations leaving aside a discussion of their solution to a later section.  
First we need to specify which dihedral symmetry we have.

\subsection{Dihedral Symmetry and Spectral Curves}\label{sec: dihedral symmetry of the spectral curve}
We have already noted that subgroups of $E(3)$ that fix the monopole centre are the space groups and we are interested when this is a dihedral
\footnote{
We use the label $D_\charge$ to refer to the dihedral group of degree $\charge$ and order $2\charge$, following the convention \cite{LMFDB}. When $\charge$ is odd then $D_{2\charge}=D_\charge\times C_2$.} group. Now there are a number of  
space groups ($D_\charge$, $D_{\charge d}$, $D_{\charge h}$, $C_{\charge v}$) isomorphic as abstract groups to a dihedral group but only $D_\charge$ lies in $\SO(3)$ with the remaining groups containing reflections. This is the dihedral group we shall focus on. 

The general centred $C_\charge$-invariant spectral curve $\mathcal{C}$ (with reality imposed) is
\begin{equation}\label{cyclick}
\eta^\charge+ \alpha_2\eta\sp{\charge-2}\zeta^2+
 \alpha_3\eta\sp{\charge-3}\zeta^3+ \ldots+ \alpha_{\charge-1}\eta\,\zeta\sp{\charge-1}+ \alpha_\charge\zeta\sp{\charge}+\beta [
\zeta\sp{2\charge}+(-1)\sp{\charge}]=0,
\end{equation}
where $\alpha_i,\beta \in\mathbb{R}$. The axis of the $C_\charge$ symmetry has been aligned via an overall rotation with the $x_3$-axis such that $C_\charge = \pangle{s}$ where
\[
s = \begin{pmatrix}
    \cos(2\pi/\charge) & -\sin(2\pi/\charge) & 0 \\ \sin(2\pi/\charge) & \cos(2\pi/\charge) & 0 \\ 0 & 0 & 1
\end{pmatrix},
\]
which acts on the coordinates of $T\mathbb{P}^1$ as (for an $\charge$-root of unity $\omega$) $s:(\zeta, \eta) \mapsto (\omega\zeta, \omega\eta)$. The curve  (\ref{cyclick}) is the $\charge:1$ unbranched cover of the hyperelliptic curve $\mathcal{C}/C_\charge$ of genus $\charge-1$,
\begin{align}\label{hyperck}
y^2&=(x^\charge+ \alpha_2 x\sp{\charge-2}+ \alpha_3 x\sp{\charge-3}+ \ldots+  \alpha_{\charge})^2 -(-1)^\charge
4\beta^2,
\end{align}
where $x=\eta/\zeta$ and $y=\beta[ \zeta\sp{\charge} - (-1)^\charge\zeta\sp{-\charge}]$. 
The reality of $\alpha_i,\beta$ also means that (\ref{cyclick})
is invariant under $\sigma:(\zeta, \eta) \mapsto (\bar\zeta, \bar\eta)$
which corresponds to the reflection 
$\sigma:\mathbf{v}\rightarrow
\diag(-1,1,1)\mathbf{v}$ as seen earlier. Therefore (\ref{cyclick})
is also invariant under $t:=\mathfrak{J}\circ\sigma: (\zeta, \eta) \mapsto (-1/\zeta, -\eta/\zeta^2)$
 and\footnote{
 Noting $rt:(\zeta, \eta) \mapsto (-\zeta, \eta)$
we obtain the point groups
   $D_{\charge}=\pangle{s,r}$, $C_{\charge v} \cong D_{\charge}=\pangle{s,t}$ and $C_{\charge h} \cong C_\charge \times C_2=\pangle{s,rt}$. The prismatic dihedral group $D_{\charge h}=
   \pangle{s,r,t}$ is obtained by adding any two of the above to the rotations, so giving the third. Abstractly $D_{\charge h}\cong D_\charge\times C_2$.
 }
 $\pangle{s,t}$ (abstractly the dihedral group $D_\charge$) is the full automorphism group. On the
 quotient curve (\ref{hyperck}) $t$ becomes the
 hyperelliptic involution $t:(x,y)\rightarrow (x,-y)$.

From Equation (\ref{eq: action of rotation on TP1 variables}) the most general order-2 rotation about an axis orthogonal to the $x_3$-axis is $r_q : (\zeta, \eta) \mapsto (-\bar{q}/(q\zeta), \eta/(q\zeta)^2)$ for some $q \in \mathbb{C}$ with $\abs{q}=1$. Assuming that this rotation is a symmetry of $\mathcal{C}$ we get the equation
\begin{equation*}
\pround{\frac{1}{q\zeta}}^{2\charge}\pbrace{\eta^\charge+ \sum_{i=2}^\charge (-\abs{q}^2)^i\alpha_i\eta\sp{\charge-i}\zeta^i+\beta [
(-\bar{q})^{2\charge}+(-q^2)\sp{\charge}\zeta^{2\charge} ]}=0.
\end{equation*}
We see that for such a rotation to be a symmetry it is certainly necessary that $(-q^2)^\charge = 1$ and so $-q^2 = \omega^m$ for some $m \in \mathbb{Z}$. Conjugating $r_q$ by $s$ is equivalent to multiplying $q$ by $\omega$, and $q$ is only defined up to a sign. As such modulo these equivalences we get $q=i$ or, in the case $\charge$ is even, $q = i\omega^{1/2}$. Taking $q=i$ we have $r_q := r:(\zeta, \eta) \mapsto (1/\zeta, -\eta/\zeta^2)$, which corresponds to the rotation 
\[
r = \begin{pmatrix}
        -1 & 0 & 0 \\ 0 & 1 & 0 \\ 0 & 0 & -1  
    \end{pmatrix}.
\]
Had we instead taken $q = i\omega^{1/2}$ we would find $s\circ r_q = r$, and so either way the dihedral symmetry group is $D_\charge=\pangle{s,r}$. That this additional rotation extending $C_\charge$ to $D_\charge$ existed was recognised in \cite{Hitchin1995}. 

\begin{remark}\label{remark: other possible rotation for even charge}
    If $\charge$ is even, then $\tilde{r}: (\zeta, \eta) \mapsto (-1/\zeta, \eta/\zeta^2)$ corresponding to the matrix $\diag(1, -1, -1)$ is also a possible symmetry, and indeed once we have either $r$ or $\tilde{r}$ as a symmetry we have both, as $s^{\charge/2}:(\zeta, \eta) \mapsto (-\zeta, -\eta)$. Further composing we see $s^{\charge/2}rt : (\zeta, \eta) \mapsto (\zeta, -\eta)$, and so even-charge monopoles with $D_\charge$ symmetry are inversion-symmetric. This is not true when $\charge$ is odd. 
\end{remark}

Now invariance of (\ref{cyclick}) under $r$ means we keep only the even terms,
\begin{equation}\label{dihk}
\eta^\charge+ \alpha_2\eta\sp{\charge-2}\zeta^2+
 \alpha_4\eta\sp{\charge-4}\zeta^4+ \ldots +\beta [
\zeta\sp{2\charge}+(-1)\sp{\charge}]=0.
\end{equation}
Let us denote this curve by $\mathcal{C}'$;
the full automorphism group of $\mathcal{C}'$ is $D_{\charge}\times C_2$.
Setting $x=\eta/\zeta$ in (\ref{dihk}) we have
\begin{align}
x^\charge+ \alpha_2 x\sp{\charge-2}+ \alpha_4 x\sp{\charge-4}+ \ldots+  \alpha_\charge+\beta [
\zeta\sp{\charge}+\zeta\sp{-\charge}]&=0, &&\charge\ \text{even},
\label{clprecursorcurve}
\\
x^\charge+ \alpha_2 x\sp{\charge-2}+ \alpha_4 x\sp{\charge-4}+ \ldots+  \alpha_{\charge-1} x+\beta [
\zeta\sp{\charge}-\zeta\sp{-\charge}]&=0, &&\charge\ \text{odd}.
\nonumber
\end{align}
If $y=\beta[ \zeta\sp{\charge} - (-1)^\charge\zeta\sp{-\charge}]$ then 
$r:(x,y)\rightarrow (-x, (-1)^{\charge-1}y)$; thus
$y$ is invariant under $r$ only for odd $\charge$, in which case it will be a function on the quotient curve $\mathcal{C}''={\mathcal{C}'}/\pangle{s,r}$; for even $\charge$ the function $v=xy$ is invariant. Thus we have curves\footnote{Both curves have genus
$\charge-1$ for there are no branch-points associated with the $x^2$
factor in the first equation, this only being introduced so as to express the curve in manifestly invariant coordinates for the subsequent quotient.
}
\begin{align*}
v^2&=x^2(x^\charge+ \alpha_2 x\sp{\charge-2}+ \alpha_4 x\sp{\charge-4}+ \ldots+  \alpha_\charge)^2
-4\beta^2 x^2&&\charge\ \text{even},\\ 
y^2&=(x^\charge+ \alpha_2 x\sp{\charge-2}+ \alpha_4 x\sp{\charge-4}+ \ldots+  \alpha_{\charge-1} x)^2
+4\beta^2 &&\charge  \  \text{odd}.
\end{align*}
Setting $\charge=2l$ or $\charge=2l-1$ for the even and odd cases of the curves then with
$u=x^2$ we have these curves covering $2:1$ the curves
\begin{align}\label{diheven2}
v^2&=u(u^l+ \alpha_2 u\sp{l-1}+ \alpha_4 u\sp{l-2}+ \ldots+  \alpha_\charge)^2
-4\beta^2u &&\charge \ \text{even},\\ \label{dihodd2}
y^2&=u(u^{l-1}+ \alpha_2 u\sp{l-2}+ \alpha_4 u\sp{l-3}+ \ldots+  \alpha_{\charge-1})^2
+4\beta^2 &&\charge\ \text{odd}.
\end{align}
The first has genus $l$ and the second has genus $l-1$. 
Under the cyclic transformation, it was shown in \cite{Braden2011}
that
$$\frac{\eta^{\charge-2}d\zeta}{\partial_\eta P}=\pi\sp\ast\left(-\frac1k \frac{x^{\charge-2}dx}{y}\right)
$$
for the curve (\ref{hyperck})
and we observe that this differential is invariant under $r$ for $\charge$ both even and odd. Further
$$
\frac{x^{\charge-2}dx}{y}=\begin{cases} \dfrac{x^{2l-2}dx}{y}=\dfrac{x^{2l-2}du}{2xy}=
\dfrac{u^{l-1}du}{2v},\\
\dfrac{x^{2l-3}dx}{y}=\dfrac{x^{2l-4}du}{2y}=\dfrac{u^{l-2}du}{2y}.
\end{cases}
$$
In each case we obtain the maximum degree in $u$ differential on the corresponding hyperelliptic curve and the work of
\cite{Braden2011} tells us the Ercolani-Sinha vector, if it exists, will pullback from one on the quotient curve.

\subsection{Dihedral Symmetry and the \secmath{A_{\charge-1}\sp{(1)}} affine Toda equations}

We have shown then that with (rotational) $D_\charge$ symmetry the spectral curve $\mathcal{C}$ covers either the curves (\ref{diheven2}) or (\ref{dihodd2})
depending on the parity of $\charge$, and in either case the Ercolani-Sinha vector reduces to one on the quotient curve. We now wish to identify these curves with spectral curves obtained by a reduction of the 
$A_{\charge-1}\sp{(1)}$ affine Toda equations. At that stage we will have the equivalent of the Sutcliffe ansatz: an ansatz for Nahm matrices in terms of solutions to the reduced affine Toda equations.

Motivated by \cite{Braden2023} we will prove in the following section that the restrictions
\begin{equation}\label{reductionconds}
    a_i^2=a_{\charge -i}^2, \quad b_i+b_{\charge +1-i}=0,
\end{equation}
of the $A_{\charge-1}\sp{(1)}$ affine Toda equations follow from dihedral symmetry. Here we shall identify the resulting reductions of the Toda equations and that the curves (\ref{diheven2}) or (\ref{dihodd2}) are the relevant spectral curves of the reduced problem.

First, we may view the first equation of (\ref{reductionconds}) as defining a diagram automorphism $\tau$ and subsequently an automorphism of
$A_{\charge-1}\sp{(1)}$ and the second equation as defining the invariant subspace. Recalling our identification $i \leftrightarrow \alpha_i$ we have $\tau(\alpha_i)=\alpha_{\charge-i}$ and see for example
\cite{Goddard1986} on how to extend this to the algebra. This automorphism fixes $a_0$ and so $\alpha_0$ and thus is of the correct form for constructing $C_l\sp{(1)}$ or $A_{2l}\sp{(2)}$. 

We have seen that we must distinguish between the cases $\charge$ even or odd and begin with the even case $\charge=2l$. It is perhaps clearest to illustrate this by means of the example $k=4$. 
Under (\ref{reductionconds}) and the sign choice\footnote{We are choosing this simply for comparison with the cited reference; we discuss the sign ambiguity in due course.} $a_1=-a_3$ the Lax matrix (\ref{flaschkalax}) becomes
$$
\begin{pmatrix}
-b_{1} \zeta  & a_{1} & 0 & -a_{0}  \,\zeta^{2}w 
\\
 -a_{1} \zeta^{2} & -b_{2} \zeta  & a_{2} & 0 
\\
 0 & -a_{2} \zeta^{2} & -b_{3} \zeta  & a_{3} 
\\
 wa_{0} & 0 & -a_{3} \zeta^{2} & -b_{4} \zeta 
\end{pmatrix}\mapsto
\begin{pmatrix}
-b_{1} \zeta  & a_{1} & 0 & -a_{0}  \,\zeta^{2}/w 
\\
 -a_{1} \zeta^{2} & -b_{2} \zeta  & a_{2} & 0 
\\
 0 & -a_{2} \zeta^{2} & b_{2} \zeta  & -a_{1} 
\\
w a_{0} & 0 & -a_{1} \zeta^{2} & b_{1} \zeta 
\end{pmatrix}
$$
and we see \cite[(4.25)]{Adler1980a} describing the $C_2\sp{(1)}$ system
(equation (2.6) of that reference is to be compared with our (\ref{reductionconds})). Equally we may make this identification directly.
Again setting $w=1$ for simplicity, the first constraint of (\ref{reductionconds}) reduces the variables to $\varphi_0$,
$\varphi_{13}:=\varphi_1=\varphi_3$, $\varphi_2$ (where in general $\varphi_l$ is not paired). Now the Toda equations (\ref{eqntoda}) take the form
\begin{equation*}
    \ddot \varphi_{0}=2e\sp{\varphi_0}-2e\sp{\varphi_{13}},\quad
    \ddot \varphi_{13}=-e\sp{\varphi_0}+2e\sp{\varphi_{13}}-e\sp{\varphi_2},\quad
    \ddot \varphi_{2}=-2e\sp{\varphi_{13}}+2e\sp{\varphi_2},
\end{equation*}
which may be written as
$
\ddot\varphi_a=\sum_b {\overline K}\sp{T}_{ab}\, e\sp{\varphi_b}
$
where $\overline K$ is the $C_l\sp{(1)}$ Cartan matrix (\ref{cn1diagram}) and we have the $C_l\sp{(1)}$ affine Toda equations.
For our present purposes having identified the equations of motion we need not present the full Hamiltonian reduction and simply note that the constraints of (\ref{reductionconds}) ensure consistency, as will be proven later. The $C_l\sp{(1)}$ Toda spectral curve of \cite{Adler1980a} or \cite[(7)]{Martinec1996} %\cite{McDaniel1992}
is then identified 
with the curve $\mathcal{C}''$ of (\ref{clprecursorcurve}).

Much of the odd case $\charge=2l-1$ follows similarly. Again
it is perhaps clearest to illustrate this by example. The $k=3$ case that leads to the Bullough-Dodd equation given in \cite{Braden2023} is 
perhaps not the most illustrative and we will instead take
$k=5$. Now the first constraint of (\ref{reductionconds}) yields the variables $\varphi_0$,
$\varphi_{14}:=\varphi_1=\varphi_4$, $\varphi_{23}=\varphi_2=\varphi_3$ and the Toda equations (\ref{eqntoda}) become
\begin{equation*}
    \ddot \varphi_{0}=2e\sp{\varphi_0}-2e\sp{\varphi_{14}},\quad
    \ddot \varphi_{14}=-e\sp{\varphi_0}+2e\sp{\varphi_{14}}-e\sp{\varphi_{23}},\quad
    \ddot \varphi_{23}=-2e\sp{\varphi_{14}}+e\sp{\varphi_{23}}.
\end{equation*}
Recalling that the diagonal elements of a a Cartan matrix are $2$ and
the  coefficient of $e\sp{\varphi_{23}}$ in the final equation is 
not an integer multiple of this we may worry we are not seeing the Toda equations. The resolution as described earlier is that
$m_l=1/2$ is non-integral for this root system with three root lengths, and upon noting this we again arrive\footnote{Forgetting the root system interpretation, simply at the level of equations we are also free to simply shift $\varphi_{23}
\mapsto \varphi_{23}+\ln2$.} at the Toda equations
$$
\ddot\varphi_a=\sum_b {\overline K}\sp{T}_{ab}\,m_b\, e\sp{\varphi_b},\qquad
{\overline K}=\begin{pmatrix}
    2&-1&0\\-2&2&-1\\0&-2&2
\end{pmatrix},
$$
where now $\overline K$ is the $A_{2l}\sp{(2)}$ Cartan matrix (\ref{a2n2diagram}). 
We have calculated the spectral curve for $A_4\sp{(2)}$ in (\ref{a42curve}) which is seen to yield the appropriate
curve $\mathcal{C}''$ of (\ref{clprecursorcurve}), and the general case follows similarly. 
For purposes of comparison
\cite[Example 9.10]{Adler2004} discusses the $A_4\sp{(2)}$ case and our Lax reproduces this.
 In this regard it is worth noting that curves for the affine Toda equations 
are usually given by the vanishing of (\ref{flaschkacurve}) as a
function of $\eta$ and the affine parameter $w$ with constant $\zeta$
while here we construct the curve in terms of $\eta$ and $\zeta$
with $w=1$.

\section{The Theorem}\label{sec: the theorem}

The theorem we shall prove is the following. 
\begin{theorem}\label{thm: main Dihedral}
Any $D_\charge$ rotationally dihedrally symmetric charge-$\charge$ Euclidean $\mathfrak{su}(2)$ BPS monopole is, up to an overall rotation, given by Nahm data gauge equivalent to (\ref{sutcliffeansatz}) obtained from the $\mathfrak{su}(\charge)$ affine Toda equations $A_{\charge-1}\sp{(1)}$ subject to the restriction of the variables (\ref{reductionconds}) arising from the folding procedure. For $k=2l$ and $k=2l-1$ these are respectively the $C_l^{(1)}$ and $A_{2(l-1)}^{(2)}$ affine Toda equations.
	% Any charge-$\charge$ monopole with rotational spatial dihedral symmetry and irreducible spectral curve is gauge equivalent to Nahm matrices of the form (\ref{sutcliffeansatz}) obtained from the $su(\charge)$ affine Toda equations $A_{\charge-1}\sp{(1)}$ subject to the restriction of the variables arising from the folding procedure.
\end{theorem}
\begin{proof}
	Recall that for a monopole to be invariant under a rotation by $A \in \SO(3)$, there must exist a constant matrix $C = \rho(A) \in \SU(\charge)$ such that the corresponding Nahm data satisfies 
	\[
	T_i = \sum_{j=1}^3 A_{ij} \psquare{C T_j C^{-1}}. 
	\]
As shown in \S\ref{sec: dihedral symmetry of the spectral curve} by rotating the monopole overall we need only consider the case of $A = r = \diag(-1, 1, -1)$, but for the sake of making an observation later we shall pick the matrix $A$ to be the more general $\diag((-1)^m, (-1)^{m-1}, -1)$ for some $m$. We shall vectorise\footnote{All we shall require about vectorisation is that it is a linear map $V$ sending matrices to vectors such that $V(ABC) = (C^T\otimes A)V(B)$.} the above matrix equation with the diagonal $A$ to get the conditions 
	\[
	\psquare{1 \otimes T_i - A_{ii} T_i^T \otimes 1}V(C) = 0 ,
	\]
	which we will rewrite as 
	\[
	\underbrace{\psquare{1 \otimes (T_1 \pm iT_2) - (-1)^m(T_1\mp iT_2)^T \otimes 1}}_{M_\pm}V(C) = 0 = \underbrace{\psquare{1 \otimes (2iT_3) + (2iT_3) \otimes 1}}_{M_3}V(C).
	\]
	% Note that if the $a_i$ are all real, as we know from \cite{Braden2011} can be achieved, the two equations $M_\pm V(C)=0$ are equivalent. 
 To write these conditions explicitly, recall in terms of elementary matrices we have 
	\begin{align*}
T_1 + i T_2 &=	a_0 E_{\charge,1} + \sum_{i=1}^{\charge-1} a_i E_{i, i+1}, \\
2 i T_3 &= \sum_{i=1}^\charge b_i  E_{i, i}. 
	\end{align*}
This gives $M_3 = \diag(m_i)$ with $m_{\charge(i-1)+j} = b_i+b_j$ for $i,j=1, \dots, \charge$. For this matrix to have nontrivial kernel giving an invertible matrix is it certainly necessary that there exists a permutation $\sigma \in S_\charge$ such that $b_{i}+b_{\sigma(i)}=0$, and this permutation will be order two. As such, we can (without loss of generality by permuting the indices) choose $\sigma(i) = \charge+1-i$. In order to have this condition remain valid over the flow it is necessary that 
\begin{equation}\label{eq: consistency of b restriction}
\dot{b}_i + \dot{b}_{\charge+1-i} = \pround{a_i^2 - a_{i-1}^2} + \pround{a_{\charge+1-i}^2 - a_{\charge-i}^2} = 0,
\end{equation}
and we shall impose this consistency condition later. 

Having imposed the restrictions on the $b_i$, we get the $\charge^2 \times \charge$ matrix
\[
K = \sum_{i=1}^\charge E_{\charge(i-1) + (\charge+1-i), i} = \sum_{i=1}^\charge E_{(\charge-1)i+1, i}
\]
whose columns span $\Ker M_3$. It is simple to compute that 
\begin{align*}
\psquare{1 \otimes (T_1 + iT_2)}K &= a_0 E_{\charge^2, \charge} + \sum_{i=1}^{\charge-1} a_{\charge-i} E_{i(\charge-1), i}, \\
\psquare{(T_1 - iT_2)^T \otimes 1}K &= [-\overline{a_0}] E_{\charge^2, 1} + \sum_{i=1}^{\charge-1} [-\overline{a_i}] E_{(\charge-1)i, i+1},
\end{align*}
noting that at present we have not assumed the $a_i$ are real-valued. As such if we write $V(C) = Kv$ for some length-$\charge$ vector $v$ we then get the equations
\begin{align*}
    a_0 v_\charge + (-1)^m \overline{a_0}v_1 &= 0, \\
    a_{\charge-i} v_i + (-1)^m \overline{a_i} v_{i+1} &= 0, \quad i=1, \dots, \charge-1.
\end{align*}
Recall now at this point we are considering when the spectral curve is irreducible. This certainly means $\beta \neq 0$, and so we cannot have any $a_i=0$ anywhere. As such we solve these equations iteratively to have $v_{i+1} = (-1)^{m-1} \frac{a_{\charge-i}}{\overline{a_i}}v_i$, yielding the compatibility condition (using $\beta \propto \prod_{i=0}^{\charge-1}a_i$)
\[
(-1)^{m-1} \frac{\overline{a_0}}{a_0}v_1 = v_\charge = (-1)^{(\charge-1)(m-1)} \frac{\prod_i a_i}{\prod_i \overline{a_i}} v_1 \Rightarrow \overline{\beta} = (-1)^{\charge(m-1)}\beta. 
\]
At this point we shall restrict to the case where we have rotated the overall monopole about the axis of the $C_\charge$-rotation so as to make the $a_i$ real-valued. Note that by doing so, we may be forced into a choice of additional order-2 rotational symmetry, we cannot just assume it. As an example of this, consider the consistency condition: with $a_i \in \mathbb{R}$ it now reduces to $\beta = (-1)^{\charge(m-1)}\beta$, and so to avoid $\beta = 0$ in the case of odd $\charge$ we must have $m$ odd, i.e. $A = \diag(-1, 1, -1)$. This tells us that $A = \diag(1, -1, -1)$ cannot be a symmetry in the $\charge$ odd case, but that it may in the $\charge$ even case. We have seen this earlier from the perspective of transforming the variables of the spectral curve in Remark \ref{remark: other possible rotation for even charge}.

Define now the ratios $\lambda_i = \frac{a_{\charge-i}}{a_i}$ for $i=1, \dots, \charge-1$, which we need to be constant in order to have $C$ be a constant matrix in $\SU(\charge)$. We see
\begin{align*}
\dot{\lambda}_i &= \frac{1}{{{a_i}^2}} \psquare{\dot{a}_{\charge-i} {a_i} - a_{\charge-i} {\dot{a}_i}} , \\
&= \frac{a_{\charge-i}}{{2{a_i}}} \psquare{(b_{\charge-i} - b_{\charge-i+1}) - (b_i - b_{i+1})}, \\
& = 0,
\end{align*}
when we have restricted to $b_{i} + b_{\charge+1-i} = 0$. This tells us that if $a_i^2 - a_{k-i}^2=0$ at some $s$, it is zero for all $s$. Finally, we wish to impose the consistency condition of Equation (\ref{eq: consistency of b restriction}), yielding 
\begin{align*}
a_i^2 (1-\lambda_i^2) - a_{i-1}^2(1 - \lambda_{i-1}^2) &= 0, \quad i=2, \dots, \charge-1, \\
a_{1}^2 (1- \lambda_1^2 ) &= 0.
\end{align*}
If for all $i$, $\lambda_i^2 = 1$, we get a solution to the above equations. If for some $i$, $\lambda_{i-1}^2 =1$ but $\lambda_i^2 \neq 1$, then $a_i = 0$ contradicting our assumption of irreducibility. As such all $\lambda_i^2=1$ as $\lambda_1^2=1$. 

What we have thus shown is that in the case where all the $a_i$ are real valued, provided we have the restrictions $a_i^2 - a_{\charge-i}^2 = 0 = b_{i} + b_{\charge+1-i}$ for $i=1, \dots, \charge-1$, then we can solve the equations 
\[
M_+ V(C) = 0 = M_3 V(C),
\]
with the solution being $V(C) = Kv$ where $v_{i+1} = (-1)^{m-1}\frac{a_{\charge-i}}{a_i}v_i := \tilde{\lambda}_i v_i$. Now because of our implicit conventions in vectorisation this gives (taking $\tilde{\lambda}_0 :=1 $) the antidiagonal matrix
\[
C = \sum_{i=1}^\charge v_i E_{i, \charge+1-i} = v_1 \sum_{i=1}^\charge \tilde{\lambda}_{i-1} E_{i, \charge+1-i}.
\]
One can check that 
\[
\det(C) = (-1)^{\floor{\charge/2}}v_i^\charge \prod_i \tilde{\lambda}_i  = (-1)^{\floor{\charge/2} + (m-1)}  \prod_i \lambda_i  v_1^\charge ,
\]
and moreover that 
\[
C C^\dagger = \abs{v_1}^2 \diag(\abs{\tilde{\lambda}_i}^2) = \abs{v_1}^2 I. 
\]
As such choosing $v_1$ appropriately we can ensure that $C \in \SU(\charge)$. Because we can achieve this with $v_1$ either pure real or imaginary this further means 
\[
M_- V(C) = \pm \overline{M_+ V(C)} = 0,
\]
and so we can solve all of the equations to find an appropriate constant $C \in \SU(\charge)$. 
\end{proof}

Changing the sign of any single $a_i$ changes the sign of $\beta$, and so clearly two choices of $\pbrace{\lambda_i}_{i=1}^{\floor{k/2}}$ which differ only for one $i$ do not give gauge equivalent monopoles, but the corresponding monopoles are related by a rotation. For odd $\charge$ changing the sign of $\beta$ just corresponds to rotating the overall monopole by angle $\pi$ about the $x_3$-axis, and so two sign changes give gauge equivalent monopoles. For even $\charge$ changing the sign of $\beta$ just corresponds to rotating the overall monopole by angle $\pi/2$ about the $x_3$-axis and so two sign changes gives gauge equivalent monopoles. As such, two choices of $\pbrace{\lambda_i}_{i=1}^{\floor{k/2}}$ give gauge equivalent monopoles if and only if they differ at an even number of $i$. Choosing all $\lambda_i = -1$ recovers the identification with \cite{Adler1980a} made earlier.

\section{Solutions}\label{sec: solutions}
We know from the rational map construction of monopoles \cite{Hitchin1995} that monopoles with dihedral symmetry must exist. This means that there must exist solutions to the affine Toda equations with pole behaviour at $s=0,2$ such that the corresponding Nahm matrices have simple poles whose residues form an irreducible representation of $\mathfrak{su}(2)$. Historically, it was the difficulty in identifying these boundary conditions that prevented Sutcliffe from finding the dihedrally symmetric 3-monopole \cite{Sutcliffe1996a}. We shall now describe a procedure to find equations for these boundary conditions and solve them.

The method we shall describe involves constructing the \emph{residue variety}. We define the matrices $X_i := -\lim_{s\to 0} (sT_i)$ which by assumption are well-defined complex matrices. These satisfy the equations 
\begin{align}
    [X_i, X_j] &= \sum_{l=1}^3 \epsilon_{ijl}X_l , \label{eq: equations for residue variety}\\
    \sum_{i=1}^3 X_i^2 &= \frac{-1}{4}(\charge^2-1)I . \label{eq: equations for residue variety 2}
\end{align}
Equation (\ref{eq: equations for residue variety}) is necessary and sufficient for the $X_i$ to give a $\charge$-dimensional representation of $\mathfrak{su}(2)$ and follows from Nahm's equations. Equation (\ref{eq: equations for residue variety 2}) is necessary for the representation to be irreducible, which follows from realising that $\sum_i X_i^2$ is the representation of a Casimir of $\mathfrak{su}(2)$ \cite{Taylor1986}. Together Equations (\ref{eq: equations for residue variety}) and (\ref{eq: equations for residue variety 2}), expanded in terms of the entries of the $X_i$, define an affine variety in $\mathbb{C}^{3\charge^2}$ which we call the residue variety. Elements of the reside variety give boundary conditions for Nahm's equations at $s=0$. 

Return now to the Nahm matrices (\ref{sutcliffeansatz}) written in terms of the Flaschka variables, introduce Laurent series expansions around $s=0$
\begin{align*}
    a_i &= \frac{A_i}{s} + \sum_{r \geq 0} a_{i,r} s^r, \\
    b_i &= \frac{B_i}{s} + \sum_{r \geq 0} b_{i,r} s^r , 
\end{align*}
and write $A_\charge = A_0$. We are able to explicitly describe the residue variety in this case. 

\begin{prop}
    Given $\charge > 1$, the residue variety is cut out by the equations  
    \begin{align*}
        A_i &= -\frac{1}{2} A_i (B_{i} - B_{i+1}), \\
        B_i &= -(A_i^2 - A_{i-1}^2), \\
        \frac{1}{4} (\charge^2-1) &= A_i^2 + \frac{1}{2} B_i + \frac{1}{4 }B_i^2. 
    \end{align*}
    Moreover, the intersection of the residue variety and the plane $A_0 = 0$ is exactly the collection of points
    \begin{align*}
        A_i^2 &= i(\charge-i), \\
        B_i &= -(\charge-2i+1). 
    \end{align*}
\end{prop}
\begin{proof}
The first two equations follow simply from the Toda equations substituting in the Laurent series expansion. To find the third equation it is easiest to note 
\[
\sum_i X_i^2 = \psquare{(X_1 + iX_2)(X_1 - iX_2) + 2iX_3} + X_3^2,
\]
from which it follows via algebra that 
\[
\sum_i X_i^2 = -\diag\psquare{A_i^2 + \frac{1}{2}B_i + \frac{1}{4} B_i^2}. 
\]

We now want to further impose $A_0=0$. This is not an artifical condition to impose: if all $A_i$ are nonzero then we have the equations
\begin{align*}
    B_{i+1} &= B_i + 2, \quad i=1, \dots, \charge-1, \\
    B_1 &= B_\charge + 2,
\end{align*}
and these are inconsistent. As such, at least one of the $A_i$ is zero, and by cycling the indices we can choose it to be $A_0$. 
Observe that we thus get
\[
B_\charge = A_{\charge-1}^2, \quad \charge^2 - 1 = 2B_\charge + B_\charge^2 \quad \Rightarrow \quad B_\charge = \charge-1. 
\]
(The alternate solution to the quadratic is $B_\charge = -\charge-1$). Note now that if we had $A_i = 0 = A_{i-1}$ for some $i$, then $B_i = 0$ by the second defining equation and then the third defining equation is inconsistent. As such we know $A_{\charge-1} \neq 0$, and so we can determine $B_{\charge-1}$ by the linear equation
\[
1 + \frac{1}{2}(B_{\charge-1} - B_\charge) = 0. 
\]
We may then determine $A_{\charge-1}$ using the third defining equation, and it will be nonzero. This allows us to iterate the procedure down to find all the $B_i$ and $A_i$, giving the desired values. 
\end{proof}

We observe that for our residues of the Flaschka variables we have
\[
A_i^2 - A_{\charge-i}^2 = 0 = B_i + B_{\charge+1-i},
\]
which are exactly the conditions required in order to have dihedrally symmetric charge-$\charge$ Nahm data.

\section{Discussion}\label{sec: discussion}

The purpose of this short work in proving Theorem \ref{thm: main Dihedral} has been to generalise the work of \cite{Sutcliffe1996a, Braden2011} to give a complete description of the Nahm matrices for charge-$\charge$ monopoles with a $D_\charge$ rotational symmetry in terms of the affine Toda equations associated to certain Lie algebras, namely $C_{l}^{(1)}$ and $A_{2l}^{(2)}$. This we were able to achieve by making clear the correspondence between folding the $A_{\charge-1}^{(1)}$ Lie algebra associated $C_\charge$-symmetric monopoles and imposing an additional order-2 rotational symmetry on these monopoles. Appealing to the McKay correspondence, Sutcliffe \cite{Sutcliffe1996a} had expected that dihedrally symmetric monopoles would be related to type-D simply laced Lie algebras, but this appears not to be the case.
Although we do not solve the reduced equations here, leaving this for a future work, we do show that the constraints (\ref{reductionconds})
allow for the irreducible representaion required of Nahm data. When we return to the finite-gap integration of our reduced equations we
hope to relate our curves to the known descriptions where Toda flows lie on certain Prym varieties related to the spectral curve (see
\cite{Adler1980a,Kanev1989,McDaniel1992}) noting again that the latter
are in terms of $\eta$ and the affine parameter $w$ which differ from our parameterizations of the curves.

A new discovery we had not predicted was that the residue variety for $C_\charge$-invariant $\charge$-monopoles completely determines the asymptotic behaviour of the Flaschka variables as $s \to 0$. If one were to find a real solution to the $A_{\charge-1}^{(1)}$ affine Toda equations with simple poles for the $a_i, b_i$ at $s=0$ they need not necessarily have the correct residues to get valid Nahm data by violating the irreducibility condition, as was seen in \cite{Sutcliffe1996a}. The same is true for $D_\charge$-invariant monopoles, but by having already imposed the conditions $A_i^2 - A_{\charge-i}^2 = 0 = B_i + B_{\charge+1-i}$ the number of remaining choices is reduced, so much so that in the $\charge=3$ case just requiring any $A_i$ is nonzero is sufficient. This suggests that aiming to solve the Nahm constraints for $D_\charge$-symmetric $\charge$-monopoles may be less difficult than for other monopoles of the same charge, owing also to the fact that we know for such curves Nahm's equations linearise on the Jacobian of a genus-$\floor{\frac{k}{2}}$ curve. As a starting point for future work, one might investigate the charge-4 $D_4$-symmetric monopoles for which the Nahm data can be solved for in terms of genus-2 hyperelliptic functions, as the $\charge=3$ case was already solved \cite{Braden2023}.

One aspect that we have not commented on in this work is the role of the principal $\SO(3)$ embedding. It is known from study of affine Toda field theory \cite{Olive1993} that there is a key interplay between this and the automorphism involved in folding; similarly the
principal subalgebra played a crucial role in the proof of \cite{Braden2011}. Clarification of this may lead to employing the
ansatz (\ref{nahmansatz}) in describing higher-charge monopoles.

%%%%%%%%%%%%%%%%%%%%%%%%%%%%%%%%%%%%%%%%%%%%
%%%%%%%%%%%%%%%%%%%%%%%%%%%%%%%%%%%%%%%%%%%%
%%%%%%%%%%%%%%%%%%%%%%%%%%%%%%%%%%%%%%%%%%%%
%%%%%%%%%%%%%%%%%%%%%%%%%%%%%%%%%%%%%%%%%%%%

%\bibliography{jabref_library}
\addcontentsline{toc}{chapter}{Bibliography}
% Choose a reference style from https://verbosus.com/bibtex-style-examples.html
\bibliographystyle{amsalpha}
% Or use the `amsrefs' package, https://ams.org/tex/amsrefs.html

\providecommand{\bysame}{\leavevmode\hbox to3em{\hrulefill}\thinspace}
\providecommand{\MR}{\relax\ifhmode\unskip\space\fi MR }
\providecommand{\MRhref}[2]{%
  \href{http://www.ams.org/mathscinet-getitem?mr=#1}{#2}
}
\providecommand{\href}[2]{#2}

\end{document}